\documentclass[11pt,english]{article}
\usepackage[T1]{fontenc}
\usepackage[latin9]{inputenc}
\usepackage{geometry}
\usepackage{amsmath}
\usepackage{amsthm}
\usepackage{amssymb}

\makeatletter
\theoremstyle{plain}
\newtheorem{thm}{\protect\theoremname}[section]
  \theoremstyle{definition}
  \newtheorem{defn}[thm]{\protect\definitionname}
  \theoremstyle{plain}
  \newtheorem{fact}[thm]{\protect\factname}
  \theoremstyle{plain}
  \newtheorem{lem}[thm]{\protect\lemmaname}
  \theoremstyle{remark}
  \newtheorem{claim}[thm]{\protect\claimname}

\usepackage{complexity,color,bm}
\usepackage{colortbl}

\providecommand{\E}{\mathrm{E}}

\newcommand{\e}{\mathrm{e}}

\definecolor{gray-comment}{gray}{0.5}

\theoremstyle{plain}

\makeatletter
\newtheorem*{rep@theorem}{\rep@title}
\newcommand{\newreptheorem}[2]{%
\newenvironment{rep#1}[1]{%
 \def\rep@title{#2 \ref{##1}}%
 \begin{rep@theorem}}%
 {\end{rep@theorem}}}
\makeatother

\newreptheorem{rcor}{Corollary}
\newreptheorem{rthm}{Theorem}

\usepackage{thmtools}
\usepackage{thm-restate}

\makeatother

\usepackage{babel}
  \providecommand{\claimname}{Claim}
  \providecommand{\definitionname}{Definition}
  \providecommand{\factname}{Fact}
  \providecommand{\lemmaname}{Lemma}
\providecommand{\theoremname}{Theorem}

\begin{document}

\title{Detecting communities is hard\\and counting them is even harder}

\author{Aviad Rubinstein}
\maketitle
\begin{abstract}
We consider the algorithmic problem of community detection in networks.
Given an undirected friendship graph $G=\left(V,E\right)$, a subset
$S\subseteq V$ is an $\left(\alpha,\beta\right)$-community if:
\begin{itemize}
\item Every member of the community is friends with an $\alpha$-fraction
of the community;{\small \par}
\item Every non-member is friends with at most a $\beta$-fraction of the
community.{\small \par}
\end{itemize}
Arora et al \cite{AroraGSS12-communities-QPTAS} gave a quasi-polynomial
time algorithm for enumerating all the $\left(\alpha,\beta\right)$-communities
for any constants $\alpha>\beta$. 

Here, we prove that, assuming the Exponential Time Hypothesis (ETH),
quasi-polynomial time is in fact necessary - and even for a much weaker
approximation desideratum. Namely, distinguishing between:
\begin{itemize}
\item $G$ contains an $\left(1,o\left(1\right)\right)$-community; and{\small \par}
\item $G$ does not contain an $\left(\beta+o\left(1\right),\beta\right)$-community
for any $\beta\in\left[0,1\right]$.{\small \par}
\end{itemize}
We also prove that counting the number of $\left(1,o\left(1\right)\right)$-communities
requires quasi-polynomial time assuming the weaker \#ETH.
\end{abstract}

\section{Introduction}

Identifying communities is a central graph-theoretic problem with
important applications to sociology and marketing (when applied to
social networks), biology and bioinformatics (when applied to protein
interaction networks), and more (see e.g. Fortunato's classic survey
\cite{Fortunato10-survey}). Defining what exactly is a {\em community}
remains an interesting problem on its own (see Arora et al \cite{AroraGSS12-communities-QPTAS}
and Borgs et al \cite{BCMT16-community-coNP-complete} for excellent
treatment from a theoretical perspective). Ultimately, there is no
single ``right'' definition, and the precise meaning of community
should be different for social networks and protein interaction networks. 

In this paper we focus on the algorithmic questions arising from one
of the simplest and most canonical definitions, which has been considered
by several theoretical computer scientists \cite{MishraSST07-a_b-cluster,AroraGSS12-communities-QPTAS,BBBCT13-communities,BKRW15-DkS}
(see Subsection \ref{sub:Related-works} for further discussion):
\begin{defn}
[$(\alpha, \beta)$-Community]

Given an undirected graph $G=\left(V,E\right)$ an $(\alpha,\beta)$-community
is a subset $S\subseteq V$ that satisfies:
\begin{description}
\item [{{Strong ties inside the community}}] For every $v\in S$, $\left|\left\{ v\right\} \times S\right|\cap E\geq\alpha\cdot\left|S\right|$;
and
\item [{{Weak ties to nodes outside the community}}] For every $u\notin S$,
$\left|\left\{ u\right\} \times S\right|\cap E\leq\beta\cdot\left|S\right|$.
\end{description}
\end{defn}
Arora et al \cite[Theorem 3.1]{AroraGSS12-communities-QPTAS} gave
a simple quasi-polynomial ($n^{O\left(\log n\right)}$) time for detecting
$\left(\alpha,\beta\right)$-communities whenever $\alpha-\beta$
is at least some positive constant. The algorithm enumerates over
$O\left(\log n\right)$-tuples of vertices. For each tuple, consider
the set of vertices that are neighbors of an $\left(\alpha+\beta\right)/2$-fraction
of the tuple; test whether this candidate set is indeed a community.

Arora et al's algorithm and analysis are very similar to related algorithms
for approximate Nash equilibrium \cite{LMM03_quasi_poly}, Densest
$k$-Subgraph \cite{Barman15-QPTAS} and Dughmi's Zero-Sum Signaling
problem \cite{CCDEHT15-quasipoly_signaling}. Recently, matching quasi-polynomial
hardness results have been proved for approximate Nash equilibrium
\cite{BKW15-best_nash,BPR15-PCP-PPAD,Rub16-Nash,DFS16-other_objectives},
Densest $k$-Subgraph \cite{BKRW15-DkS,Man16-DkS}, and Zero-Sum Signaling
\cite{Rub15-signaling,BCKS15-signaling} using or inspired by the
technique of ``birthday repetition'' \cite{AIM14-birthday}. A natural
question, made explicit in \cite{BKRW15-DkS}, is whether similar
techniques can be shown to prove quasi-polynomial time hardness, assuming
the Exponential Time Hypothesis (ETH)\footnote{The Exponential Time Hypothesis (ETH) \cite{IPZ01-ETH} asserts that
solving 3SAT requires time $2^{\Omega\left(n\right)}$. Note that
(given our current understanding of complexity) this assumption is
essentially necessary - an \NP-hardness result is very unlikely given
\cite{AroraGSS12-communities-QPTAS}'s quasi-polynomial algorithm.
Recall also that ETH is a significantly weaker assumption than the
related SETH \cite{IP01-SETH,CIP09-SETH2} and NSETH \cite{CGIMPS16-NSETH},}, for $(\alpha,\beta)$-community detection, for any constants $\alpha>\beta\in\left[0,1\right]$. 

Here we show that, for {\em every} constants $\alpha>\beta\in(0,1]$,
community detection requires quasi-polynomial time (assuming ETH).
For example, when $\alpha=1$ and $\beta=0.01$, this means that we
can hide a clique $C$, such that every single vertex not in $C$
is connected to at most 1\% of $C$. Our main result is actually a
much stronger inapproximability: even in the presence of a $\left(1,o\left(1\right)\right)$-community,
finding any $\left(\beta+o\left(1\right),\beta\right)$-community
is hard.
\begin{thm}
\label{thm:decision-1}For every $n$ there exists an $\epsilon=\epsilon\left(n\right)=o\left(1\right)$
such that, assuming ETH, distinguishing between the following requires
time $n^{\tilde{\Omega}\left(\log n\right)}$:
\begin{description}
\item [{Completeness}] $G$ contains an $\left(1,\epsilon\right)$-community;
and
\item [{Soundness}] $G$ does not contain an $\left(\beta+\epsilon,\beta\right)$-community
for any $\beta\in\left[0,1\right]$.
\end{description}
\end{thm}
Unlike all quasi-polynomial approximation schemes mentioned above,
Arora et al's algorithm has the unique property that it can also {\em exactly
count} all the $\left(\alpha,\beta\right)$-communities. Our second
result is that counting even the number of $\left(1,o\left(1\right)\right)$-communities
requires quasi-polynomial time. A nice feature of this result is that
we can base it on the much weaker \#ETH assumption, which asserts
that counting the satisfying assignment for a 3SAT instance requires
time $2^{\Omega\left(n\right)}$. (Note, for example, that \#ETH is
likely to be true even if $\P=\NP$.)
\begin{thm}
\label{thm:counting-1}For every $n$ there exists an $\epsilon=\epsilon\left(n\right)=o\left(1\right)$
such that, assuming \#ETH, counting $\left(1,\epsilon\right)$-communities
requires time $n^{\log^{1-o\left(1\right)}n}$.
\end{thm}

\subsection{Related works \label{sub:Related-works}}

The most closely related work is a reduction by Balcan, Borgs, Braverman,
Chayes, and Teng \cite[Theorem 5.3]{BBBCT13-communities} from Planted
Clique to finding $\left(1,1-\gamma\right)$-communities, for some
small (unspecified) constant $\gamma>0$. Note that our inapproximability
in Theorem \ref{thm:decision-1} is much stronger in all parameters;
furthermore, although formally incomparable, our ETH assumption is
preferable over the average-case hardness assumption of Planted Clique.

\subsubsection*{Algorithms for special cases}

Mishra, Schreiber, Stanton, and Tarjan \cite{MishraSST07-a_b-cluster}
gave a polynomial-time algorithm for finding $\left(\alpha,\beta\right)$-communities
that contain a vertex with very few neighbors outside the community.
Balcan et al \cite{BBBCT13-communities} give a polynomial-time algorithm
for enumerating $\left(\alpha,\beta\right)$-communities in the special
case where the degree of every node is $\Omega\left(n\right)$. 

Arora, Ge, Sachdeva, and Schoenebeck \cite{AroraGSS12-communities-QPTAS}
consider several semi-random models where the edges inside the community
are generated at random, according to the expected degree model. (In
fact, their quasi-polynomial time algorithm is also stated in this
setting, but only their ``Gap Assumption'', which is equivalent
to $\alpha-\beta=\Omega\left(1\right)$, is used in the analysis.)

\subsubsection*{Stochastic Block Model}

Variants of the community detection problem on graphs generated by
different stochastic models are extremely popular (see e.g. \cite{BMNN16-block-model_info,CKST16-local_communities,FP16-block-model_bipartite,MMV16-block-model_semi,MPW16-block-model_semi,MX16-block-model,TPGV16-clustering}
for papers in conference proceedings from June 2016). Perhaps the
most influential is the {\em Stochastic Block Model} \cite{HLL83-stochastic_blockmodel}:
The graph is partitioned into two disjoint communities; the edges
within each community are present with probability $\alpha$, independently,
whereas edges between communities are present with probability $\beta$.
Hence this model can also be seen as a special case of the $\left(\alpha,\beta\right)$-Community
Detection problem.

Stochastic models are extremely helpful in physics, for example, because
atoms' interactions obey simple mathematical formulas with high precision.
Unfortunately, for applications such as social networks, existing
models do not describe human behavior with atomic precision, hence
casting a shadow over the applicability of algorithms that work on
ideal stochastic models. Recent works \cite{MPW16-block-model_semi,MMV16-block-model_semi}
attempted to bridge the gap from ideal model to practice by showing
that certain SDP-based algorithms continue to work in a particular
semi-random model where a restricted adversary is allowed to modify
the random input graph. These success stories beg the question of
how strong can one make the adversary? The current paper illuminates
some of the computational barriers.

\subsubsection*{Alternative approaches to modeling communities}

As we mentioned above, there are many different definitions of ``communities''
in networks. For in-depth discussion of different definitions see
Arora et al \cite{AroraGSS12-communities-QPTAS} or Borgs et al \cite{BCMT16-community-coNP-complete}.
As pointed out by the latter, for some definitions even verifying
that a candidate subset is a community is intractable. 

There is also an important literature on axiomatic approaches to the
related problem of clustering (e.g. \cite{Kleinberg02-clustering_Arrow,BA08-clustering_quality,LM14-clustering_axioms});
note that while clustering typically aims to partition a set of nodes,
our main focus is on detecting just a single community; in particular,
different communities may intersect.

\subsection{Overview of proofs}

A good starting point for the technical discussion is a recent subexponential
reduction from 3SAT to the related problem of {\sc Densest-$k$-Subgraph}
\cite{BKRW15-DkS}. In {\sc Densest-$k$-Subgraph}, we seek a subgraph
of size $k$ of maximal density. The two ingredients in \cite{BKRW15-DkS}'s
reduction are ``birthday repetition'' \cite{AIM14-birthday} and
the ``FGLSS graph'' \cite{fglss96}:
\begin{description}
\item [{{``Birthday repetition''}}] Starting with an instance of {\sc Label Cover}
(see definition in Section \ref{sec:Preliminaries}), the reduction
considers a mega-variable for every $\rho$-tuple of variables, for
$\rho\approx\sqrt{n}$. By the birthday paradox, almost every pair
of $\rho$-tuples of variables intersect, inducing a consistency constraint
on the two mega-assignments. Similarly, we expect to see some {\sc Label Cover}
edges in the union of the two $\rho$-tuples, inducing an additional
{\sc Label Cover} constraint between the two mega assignments. Notice
that we have ${n \choose \rho}\approx2^{\sqrt{n}}$ mega variables,
and the alphabet size is also approximately $N=2^{\sqrt{n}}$. Therefore,
assuming ETH, finding an approximately satisfying assignment for the
mega-variables requires time $2^{\Omega\left(n\right)}\approx N^{\log N}$.
\item [{FGLSS}] Similarly to the classic reduction by Feige et al. \cite{fglss96}
for the {\sc Clique} problem, \cite{BKRW15-DkS} construct a vertex
for each mega assignment to each mega variable, and draw an edge between
two vertices if the induced assignments do not violate any consistency
or {\sc Label Cover} constraints. Notice that if the {\sc Label
Cover} instance has a satisfying assignment, then the graph contains
a clique of size ${n \choose \rho}$ where each mega variable receives
the mega assignment induced by the globally satisfying assignment.
On the other hand, any subgraph that corresponds to a consistent assignment
which violates many constraints must be missing most of its edges.
\end{description}
Now this simple reduction is still far from working for the {\sc Community
Detection} problem, and indeed the latter was listed as an open problem
in \cite{BKW15-best_nash}. Below we describe some of the obstacles
and outline how we overcome them.

\subsubsection*{Completeness}

Surprisingly, the main problem with using the same reduction for {\sc Community
Detection} is the completeness: even if the {\sc Label Cover} instance
has a satisfying assignment, the resulting graph has no $\left(\alpha,\beta\right)$-communities,
for any constants $\alpha>\beta$! Observe, in particular, that the
clique that corresponds to the satisfying assignment does not satisfy
the weak ties condition. For any vertex $v$ in that clique, consider
any vertex $v'$ that corresponds to changing the assignment to just
one variable $x_{i}$ in $v$'s assignment. If $v$ agrees with the
assignments of all other vertices in the clique, $v'$ agrees with
almost all of them - except for the negligible fraction that cover
$x_{i}$ or its neighbors in the {\sc Label Cover} graph.

To overcome this problem of vertices that are ``just outside the
community'', we use error correcting codes. Namely, we encode each
assignment as a low-degree bivariate polynomial over finite field
${\cal G}$ of size $\left|{\cal G}\right|\approx\sqrt{n}$. Now vertices
correspond to low-degree assignments to rows/columns of the polynomial.
This guarantees that the assignments induced by every two vertices
are far. If $v$ agrees with all other vertices in the community,
then almost all of those vertices disagree with $v'$.

\subsubsection*{Soundness}

The main challenge for soundness is ruling out communities that do
not correspond to a single, globally consistent assignment to the
{\sc Label Cover} instance. The key idea is to introduce auxiliary
vertices that punish such communities by violating the weak ties desideratum.

Let us begin with the reduction to the counting variant (Theorem \ref{thm:counting-1}),
which is easier, mostly because we are not concerned with approximation
(i.e. we only have to show that subsets that are exactly $\left(1,\epsilon\right)$-communities
correspond to satisfying assignments). Here we further simplify matters
by sketching a construction with weighted edges. The full reduction
(Section \ref{sec:Hardness-of-Counting}) uses unweighted edges and
is only slightly more involved. Consider, for every $g\in{\cal G}$,
an auxiliary vertex that is $\epsilon$-connected to all proper vertices
that do not correspond to assignments to the $g$-th row/column. Now
if a $\left(1,\epsilon\right)$-community $C$ does not contain a
vertex with assignment to the $g$-th row/column, the auxiliary vertex
must simultaneously: (i) belong to $C$ so as not to violate the weak
ties desideratum; yet (ii) it cannot belong to $C$ because all its
edges have weight $\epsilon$ (this would violate the strong ties
desideratum). Therefore every $\left(1,\epsilon\right)$-community
assigns values to every row/column in ${\cal G}^{2}$.

The reduction we described above suffices to show that (assuming ETH)
deciding whether the graph contains a $\left(1,\epsilon\right)$-community
also requires quasi-polynomial time. To get the stronger statement
of Theorem \ref{thm:decision-1} we must rule out even $\left(\beta,\beta+\epsilon\right)$-communities
in case the {\sc Label Cover} instance is far from satisfiable.
In particular, we need to show that subsets that do not correspond
to unique, consistent assignments are never $\left(\beta,\beta+\epsilon\right)$-communities.
Instead of a single column/row, we let each proper vertex correspond
to a subset of $t\approx\log n$ columns/rows. Instead of a single
$g\in{\cal G}$, each auxiliary vertex corresponds to subset $H\subset{\cal G}$
of size $\left|H\right|=\left|{\cal G}\right|/2$. We draw an edge
between an auxiliary vertex and a proper vertex if the indices of
all $t$ columns/rows are contained in $H$; if they are picked randomly
this only happens with polynomially small probability. If, however,
a $\beta$-fraction of the community is restricted to a small subset
$R\subset{\cal G}$, then there are auxiliary vertices for $H\supseteq R$
that connect to all those nodes and violate the weak ties desideratum.
Roughly, we show that at least a $\left(1-\beta\right)$-fraction
of the vertices have assignments that are ``well spread'' over ${\cal G}^{2}$,
and among those assignments there are many violations of the {\sc Label
Cover} constraints.

\section{\label{sec:Preliminaries}Preliminaries}

\subsection*{Label Cover}
\begin{defn}
[{\sc Label Cover}]

{\sc Label Cover} is a maximization problem. The input is a bipartite
graph $G=\left(A,B,E\right)$, alphabets $\Sigma_{A},\Sigma_{B}$,
and a projection $\pi_{e}:\Sigma_{A}\rightarrow\Sigma_{B}$ for every
$e\in E$. 

The output is a labeling $\varphi_{A}:A\rightarrow\Sigma_{A}$, $\varphi_{B}:B\rightarrow\Sigma_{B}$.
Given a labeling, we say that a constraint (or edge) $\left(a,b\right)\in E$
is {\em satisfied} if $\pi_{\left(a,b\right)}\left(\varphi_{A}\left(a\right)\right)=\varphi_{B}\left(b\right)$.
The {\em value of a labeling} is the fraction of $e\in E$ that
are satisfied by the labeling. The value of the instance is the maximum
fraction of constraints satisfied by any assignment.\end{defn}
\begin{thm}
[{Moshkovitz-Raz PCP \cite[Theorem 11]{MR10-2query-PCP}}]\label{thm:label-cover}

For every $n$ and every $\epsilon>0$ (in particular, $\epsilon$
may be a function of $n$), solving {\sc 3SAT} on inputs of size
$n$ can be reduced to distinguishing between the case that a $\left(d_{A},d_{B}\right)$-bi-regular
instance of {\sc Label Cover}, with parameters $\left|A\right|+\left|B\right|=n^{1+o\left(1\right)}\cdot\poly\left(1/\epsilon\right)$,
$\left|\Sigma_{A}\right|=2^{\poly\left(1/\epsilon\right)}$, and $d_{A},d_{B},\left|\Sigma_{B}\right|=\poly\left(1/\epsilon\right)$,
is completely satisfiable, versus the case that it has value at most
$\epsilon$.
\end{thm}
Counting the number of satisfying assignments is even harder. The
following hardness is well-known, and we sketch its proof only for
completeness:
\begin{fact}
There is a linear-time reduction from \#3SAT to counting the number
of satisfying assignments of a {\sc Label Cover} instance.\end{fact}
\begin{proof}
Construct a vertex in $A$ for each variable and a vertex in $B$
for each clause. Set $\Sigma_{A}\triangleq\left\{ 0,1\right\} $ and
let $\Sigma_{B}\triangleq\left\{ 0,1\right\} ^{3}\setminus\left(000\right)$
(i.e. $\Sigma_{B}$ is the set of satisfying assignments for a 3SAT
clause, after applying negations). Now if variable $x$ appears in
clause $C$, add a constraint that the assignments to $x$ and $C$
are consistent (taking into account the sign of $x$ in $C$). Notice
that any assignment to $A$: (i) corresponds to a unique assignment
to the 3SAT formula; and (ii) if the 3SAT formula is satisfied, this
assignment uniquely defines a satisfying assignment to $B$. Therefore
there is a one-to-one correspondence between satisfying assignments
to the 3SAT formula and to the instance of {\sc Label Cover}.
\end{proof}

\subsection*{Finding a good partition}
\begin{thm}
[{$k$-wise independence Chernoff bound \cite[Theorem 5.I]{SSS95-chernoff-k-wise}}]\label{thm:chernoff-k-wise}
Let $x_{1}\dots x_{n}\in\left[0,1\right]$ be $k$-wise independent
random variables, and let $\mu\triangleq\E\left[\sum_{i=1}^{n}x_{i}\right]$
and $\delta\leq1$. Then
\[
\Pr\left[\left|\sum_{i=1}^{n}x_{i}-\mu\right|>\delta\mu\right]\leq\e^{-\Omega\left(\min\left\{ k,\delta^{2}\mu\right\} \right)}\mbox{.}
\]

\end{thm}
We use Chernoff bound with $\Theta\left(\log n\right)$-wise independent
variables to deterministically partition variables into subsets of
cardinality $\approx\sqrt{n}$. Our (somewhat naive) deterministic
algorithm for finding a good partition takes quasi-polynomial time
($n^{O\left(\log n\right)}$), which is negligible with respect to
the sub-exponential size ($N=2^{\tilde{O}\left(\sqrt{n}\right)}$)
of our reduction\footnote{Do not confuse this with the quasi-polynomial lower bound ($N^{\tilde{O}\left(\log N\right)}$)
we obtain for the running time of the community detection problem.}.
\begin{lem}
\label{lem:partition}Let $G=\left(A,B,E\right)$ be a bipartite $\left(d_{A},d_{B}\right)$-bi-regular
graph, and let $n_{A}\triangleq\left|A\right|$, $n_{B}\triangleq\left|B\right|$;
set also $n\triangleq n_{B}+n_{A}$ and $\rho\triangleq\sqrt{n}\log n$.
Let $T_{1},\dots,T_{n_{B}/\rho}$ be an arbitrary partition of $B$
into disjoint subsets of size $\rho$. There is a quasi-polynomial
deterministic algorithm (alternatively, linear-time randomized algorithm)
that finds a partition of $A$ into $S_{1},\dots,S_{n_{A}/\rho}$,
such that:
\begin{equation}
\forall i\,\,\,\bigg|\left|S_{i}\right|-\rho\bigg|<\rho/2,\label{eq:S_i}
\end{equation}
and
\begin{equation}
\forall i,j\,\,\,\Bigg|\left|\left(S_{i}\times T_{j}\right)\cap E\right|-\frac{d_{A}\rho^{2}}{n_{B}}\Bigg|<\frac{d_{A}\rho^{2}}{2n_{B}}\mbox{.}\label{eq:S_iXT_j}
\end{equation}
\end{lem}
\begin{proof}
Suppose that we place each $a\in A$ into a uniformly random $S_{i}$.
By Chernoff bound and union bound, (\ref{eq:S_i}) and (\ref{eq:S_iXT_j})
hold with high probability. Now, by Chernoff Bound for $k$-wise independent
variables (Theorem \ref{thm:chernoff-k-wise}), it suffices to partition
$A$ using a $\mbox{\ensuremath{\Theta\left(\log n\right)}}$-wise
independent distribution. Such distribution can be generated with
a sample space of $n^{O\left(\log n\right)}$ (e.g. \cite{ABI86-k-wise-construction}).
Therefore, we can enumerate over all possibilities in quasi-polynomial
time. By the probabilistic argument, we will find at least one partition
that satisfies (\ref{eq:S_i}) and (\ref{eq:S_iXT_j}).
\end{proof}

\section{Hardness of Counting Communities \label{sec:Hardness-of-Counting}}
\begin{thm}
\label{thm:counting}There exists an $\epsilon\left(n\right)=o\left(1\right)$
such that, assuming \#ETH, counting $\left(1,\epsilon\right)$-communities
requires time $n^{\log^{1-o\left(1\right)}n}$.
\end{thm}

\subsubsection*{Construction}

Begin with an instance $\left(A,B,E,\pi\right)$ of {\sc Label Cover}
of size $n=n_{A}+n_{B}$ where $n_{A}\triangleq\left|A\right|$ and
$n_{B}\triangleq\left|B\right|$. Let ${\cal G}$ be a finite field
of size $\sqrt{n}/\epsilon^{3}$, and let ${\cal F}\subset{\cal G}$
be an arbitrary subset of size $\left|{\cal F}\right|=\sqrt{n}$.
We identify between $A\cup B$ and points in ${\cal F}^{2}$; we also
identify between a subset of ${\cal G}$ and $\Sigma_{A}\cup\Sigma_{B}$.
Thus there is a one-to-one correspondence between a subset of assignments
to $P_{{\cal F}}\colon{\cal F}^{2}\rightarrow{\cal G}$ and assignments
to the {\sc Label Cover} instance. We can extend any such $P_{{\cal F}}$
to an individual-degree-$\left(\left|{\cal F}\right|-1\right)$ polynomial
$P:{\cal G}^{2}\rightarrow{\cal G}$. In the other direction, we think
of each low individual degrees polynomial $P:{\cal G}^{2}\rightarrow{\cal G}$
as a (possibly invalid) assignment to the {\sc Label Cover} instance.

For every $g\in{\cal G}$, and degree-$\left(\left|{\cal F}\right|-1\right)$
polynomials $p_{1},p_{2}:{\cal G}\rightarrow{\cal G}$ such that $p_{1}\left(g\right)=p_{2}\left(g\right)$,
we construct $1/\epsilon$ vertices $\left\{ v_{g,p_{1},p_{2},i}\right\} _{i=1}^{1/\epsilon}\subset V$
in the communities graph. Each vertex naturally induces an assignment
($p_{1},p_{2}$) on $\left({\cal G}\times\left\{ g\right\} \right)\cup\left(\left\{ g\right\} \times{\cal G}\right)$.
We draw an edge between two vertices in $V$ if they agree on the
intersection of their lines, and if their induced assignments satisfy
all the {\sc Label Cover} constraints. 

For every $g\in{\cal G}$ and $i\in\left[1/\epsilon\right]$, we also
add two identical auxiliary vertices $u_{g,i}$ which are connected
to every $v_{g',p_{1},p_{2},i}$ for $g'\neq g$ (but not to each
other).

\subsubsection*{Completeness}

For each assignment to the {\sc Label Cover} instance, we construct
a $\left(1,\epsilon\right)$-community by taking the induced assignment
$P_{{\cal F}}\colon{\cal F}^{2}\rightarrow{\cal G}$ and extending
it to an individual-degree-$\left(\left|{\cal F}\right|-1\right)$
polynomial $P:{\cal G}^{2}\rightarrow{\cal G}$. Let $C$ be all the
vertices $v_{g,p_{1},p_{2},i}$ such that $p_{1},p_{2}$ are the restrictions
of $P$ to $\left({\cal G}\times\left\{ g\right\} \right),\left(\left\{ g\right\} \times{\cal G}\right)$.
This correspondence is one-to-one and we need to show that the resulting
$C$ is actually a $\left(1,\epsilon\right)$-community. 

Because all the vertices correspond to a consistent satisfying assignment,
$C$ is a clique. Let $v_{g,q_{1},q_{2},i}\notin C$; wlog $q_{1}$
disagrees with the restriction of $P$ to $\left({\cal G}\times\left\{ g\right\} \right)$.
Since both $q_{1}$ and the restriction of $P$ are degree-$\left(\left|{\cal F}\right|-1\right)$
polynomials, they must disagree on all but at most $\left(\left|{\cal F}\right|-1\right)$
elements of ${\cal G}$. For all other $h\in{\cal G}$, the vertex
$v_{g,q_{1},q_{2},i}$ does not share edges with any $v_{h,p_{1},p_{2},j}\in C$.
Therefore, $v_{g,q_{1},q_{2},i}$ has edges to less than an $\left(\left|{\cal F}\right|/\left|{\cal G}\right|\right)$-fraction
of vertices in $C$. Finally, every auxiliary vertex $u_{g,i}$ has
edges to a $\frac{\left|{\cal G}\right|-1}{\left|{\cal G}\right|}\cdot\epsilon<\epsilon$-fraction
of the vertices in $\epsilon$. Therefore, $C$ is a $\left(1,\epsilon\right)$-community.

\subsection{Soundness}

\subsubsection*{Structure of $\left(1,\epsilon\right)$-communities}
\begin{claim}
\label{claim:structure}Every $\left(1,\epsilon\right)$-community
$C$ contains exactly $1/\epsilon$ vertices $\left\{ v_{g,p_{1},p_{2},i}\right\} _{i=1}^{1/\epsilon}$
for each $g$.\end{claim}
\begin{proof}
First, observe that $C$ cannot contain any auxiliary vertices: if
$C$ contains one copy of $u_{g,i}$, it must also contain the other;
but they don't have an edge between them, so they cannot both belong
to a $\left(1,\epsilon\right)$-community.

Now, assume by contradiction that for some $g\in{\cal G}$, $C$ does
not contain any vertices with assignments for $\left({\cal G}\times\left\{ g\right\} \right)\cup\left(\left\{ g\right\} \times{\cal G}\right)$.
Then every vertex in $C$ is connected to (both copies of) $u_{g,i}$,
for some $i\in\left[1/\epsilon\right]$. Therefore there is at least
one $i\in\left[1/\epsilon\right]$ such that $u_{g,i}$ is connected
to an $\epsilon$-fraction of the vertices in $C$. But this is a
contradiction since $u_{g,i}\notin C$.

If we ignore the auxiliary vertices (which, as we argued, $C$ does
not contain), the different vertices $v_{g,p_{1},p_{2},i}$ that correspond
to the same assignment to the same lines (i.e. if we only change $i$)
are indistinguishable. Therefore if $C$ contains one of them, it
must contain all of them (hence, at least $1/\epsilon$ vertices for
each $g$).

Finally, since $C$ is a clique, it cannot contain vertices that disagree
on any assignments. (In particular, it cannot contain more than $1/\epsilon$
vertices for each $g$.)
\end{proof}

\subsubsection*{Completing the proof}
\begin{proof}
[Proof of Soundness] By Claim \ref{claim:structure}, every $\left(1,\epsilon\right)$-community
$C$ contains exactly $1/\epsilon$ vertices $\left\{ v_{g,p_{1},p_{2},i}\right\} _{i=1}^{1/\epsilon}$
for each $g$. Furthermore, since $C$ is a clique, all the induced
assignments agree on all the intersections. So every $\left(1,\epsilon\right)$-community
corresponds to a unique consistent assignment to the {\sc Label Cover}
instance. Finally, appealing again to the fact that $C$ is a clique,
this assignment must also satisfy all the {\sc Label Cover} constraints. 
\end{proof}

\section{Hardness of Detecting Communities \label{sec:Hardness-of-Detecting}}
\begin{thm}
\label{thm:decision}There exists an $\epsilon\left(n\right)=o\left(1\right)$
such that, assuming ETH, distinguishing between the following requires
time $n^{\tilde{\Omega}\left(\log n\right)}$:
\begin{description}
\item [{Completeness}] $G$ contains an $\left(1,\epsilon\right)$-community;
and
\item [{Soundness}] $G$ does not contain an $\left(\beta+\epsilon,\beta\right)$-community
for any $\beta\in\left[0,1\right]$.
\end{description}
\end{thm}
The rest of this section is devoted to the proof of Theorem \ref{thm:decision}.
Our starting point is the {\sc Label Cover} of Moshkovitz-Raz (Theorem
\ref{thm:label-cover}). We compose the birthday repetition technique
of \cite{AIM14-birthday} with a bi-variate low-degree encoding. We
then encode this as a graph a-la FGLSS \cite{fglss96}. We add auxiliary
vertices to ensure that any $\left(\beta+\epsilon,\beta\right)$-community
corresponds, approximately, to a uniform distribution over the variables.

\subsubsection*{Construction}

Begin with a $\left(d_{A},d_{B}\right)$-bi-regular instance $\left(A,B,E,\pi\right)$
of {\sc Label Cover} of size $n=n_{A}+n_{B}$ where $n_{A}\triangleq\left|A\right|$
and $n_{B}\triangleq\left|B\right|$. Let $\rho\triangleq\sqrt{n}\log n$;
let ${\cal G}$ be a finite field of size $\rho/\epsilon^{3}=\tilde{O}\left(\rho\right)$,
and let ${\cal F}\subset{\cal G}$ be an arbitrary subset of size
$\left|{\cal F}\right|=2\rho$. Let ${\cal F}_{A},{\cal F}_{B}\subset{\cal F}$
be disjoint subsets of size $n_{A}/\rho$, $n_{B}/\rho$, respectively.
By Lemma \ref{lem:partition}, we can partition $A$ and $B$ into
subsets $X_{1},\dots,X_{\left|{\cal F}_{A}\right|}$ and $Y_{1},\dots,Y_{\left|{\cal F}_{B}\right|}$
of size at most $\left|{\cal F}\right|$ such that between every two
subsets there are approximately $\frac{d_{A}\rho^{2}}{n_{B}}=\frac{d_{B}\rho^{2}}{n_{A}}$
constraints. For $i\in{\cal F}_{A}$, we think of the points $\left\{ i\right\} \times{\cal F}\subset{\cal G}^{2}$
as representing assignments to variables in  $X_{i}$; for $j\in{\cal F}_{B}$,
we think of ${\cal F}\times\left\{ j\right\} \subset{\cal G}^{2}$
as representing assignments to variables in $Y_{j}$. Notice that
each point in ${\cal F}^{2}$ may represent an assignment to both
a vertex from $A$ and a vertex from $B$, to one of them, or to neither.
In particular, any assignment $P\colon{\cal G}^{2}\rightarrow{\cal G}$
induces an assignment for the {\sc Label Cover} instance; note that
since $\left|{\cal G}\right|>\left|\Sigma_{A}\right|\left|\Sigma_{B}\right|$,
one value $P\left(f_{1},f_{2}\right)\in{\cal G}$ suffices to describe
assignments to both $a\in A$ and $b\in B$.

Let $t\triangleq\log n\cdot\left(\frac{\left|{\cal G}\right|}{\left|{\cal F}_{A}\right|}+\frac{\left|{\cal G}\right|}{\left|{\cal F}_{B}\right|}\right)=\polylog\left(n\right)$.
We say that a subset $S\in{{\cal G} \choose t}$ is {\em balanced}
if: $\left|S\cap{\cal F}_{A}\right|=\frac{\left|{\cal F}_{A}\right|}{\left|{\cal G}\right|}\cdot t$
and $\left|S\cap{\cal F}_{B}\right|=\frac{\left|{\cal F}_{B}\right|}{\left|{\cal G}\right|}\cdot t$.
For every balanced subset $S$, consider $2t$ polynomials $q_{\ell}\colon{\cal G}\rightarrow{\cal G}$
of degree at most $\left|{\cal F}\right|-1$, representing an assignment\footnote{We will only consider polynomials that correspond to a consistent
assignment $Q$; i.e. for each point in $S\times S$ we expect the
two corresponding polynomials to agree with each other.} $Q\colon\left(S\times{\cal G}\right)\cup\left({\cal G}\times S\right)\rightarrow{\cal G}$.
For balanced $S$ and $2t$-tuple of polynomials $\left(q_{\ell}\right)$,
we construct a corresponding vertex $v_{S,\left(q_{\ell}\right)}$
in the communities graph. Let $V$ denote the set of vertices defined
so far. For $g\in{\cal G}$ we abuse notation and say that $g\in v_{S,\left(q_{\ell}\right)}$
if $g\in S$. We construct an edge in the communities graph between
two vertices in $V$ if their assignments agree on the variables in
their intersection, and their induced assignments to $A\cup B$ satisfy
all the {\sc Label Cover} constraints.

Additionally, for every $H\subset{\cal G}$ of size $\left|H\right|=\left|{\cal G}\right|/2$,
define $\left|V\right|^{2}$ identical auxiliary vertices $u_{H}$
in the communities graph. We draw an edge between auxiliary vertex
$u_{H}$ and vertex $v_{S,\left(q_{\ell}\right)}$ if $S\subset H$.
Similarly, for every $H_{A}\subset{\cal F}_{A}$ of size $\left|H_{A}\right|=\left|{\cal F}_{A}\right|/2$,
we define $\left|V\right|^{2}$ identical auxiliary vertices $u_{H_{A}}$
with edges to every vertex $v_{S,\left(q_{\ell}\right)}$ such that
$\left(S\cap{\cal F}_{A}\right)\subset H_{A}$. For $H_{B}\subset{\cal F}_{B}$
of size $\left|H_{B}\right|=\left|{\cal F}_{B}\right|/2$, we draw
edges between $u_{H_{B}}$ and $v_{S,\left(q_{\ell}\right)}$ such
that $\left(S\cap{\cal F}_{B}\right)\subset H_{B}$.

\subsubsection*{Completeness}

Suppose that the {\sc Label Cover} instance has a satisfying assignment.
Let ${\cal Z}\subseteq{\cal G}^{2}$ denote the subset of points that
correspond to at least one variable in $A$ or $B$. Let $P_{{\cal Z}}:{\cal Z}\rightarrow{\cal G}$
be the induced function on ${\cal Z}$ that corresponds to the satisfying
assignment, and let $P:{\cal G}^{2}\rightarrow{\cal G}$ be the extension
of $P_{{\cal Z}}$ by setting $P\left(f_{1},f_{2}\right)=0$ for $\left(f_{1},f_{2}\right)\in{\cal F}^{2}\setminus{\cal Z}$
(this choice is arbitrary), and then extending to an $\left(\left|{\cal F}\right|-1\right)$-individual-degree
polynomial over all of ${\cal G}^{2}$. 

Let $C$ be the set of vertices that correspond to restrictions of
$P$ to balanced sets, i.e.
\[
C=\left\{ v_{S,\left(P\mid_{S}\right)}:\,\mbox{\ensuremath{S}\,\,is balanced}\right\} ,
\]
where $P\mid_{S}$ denotes the restriction of $P$ to $\left(S\times{\cal G}\right)\cup\left({\cal G}\times S\right)$.
Since all those vertices correspond to a consistent satisfying assignment,
$C$ is a clique. 

For any vertex $v_{S,\left(q_{\ell}\right)}\notin C$, at least one
of the polynomials, $q_{\ell^{*}}$ disagrees with the restriction
of $P$ to the corresponding line. Since both $q_{\ell^{*}}$ and
the restriction of $P$ to that line are degree-$\left(\left|{\cal F}\right|-1\right)$
polynomials, they must disagree on at least $\left(1-\frac{\left|{\cal F}\right|}{\left|{\cal G}\right|}\right)$-fraction
of the coordinates. The probability that a random balanced set $S'$
is contained in the $O\left(\epsilon^{3}\right)$-fraction of coordinates
where they do agree is smaller than $\epsilon$ (and in fact polynomially
small in $n$). Therefore $v_{S,\left(q_{\ell}\right)}$ has inconsistency
violations with all but (less than) an $\epsilon$-fraction of the
vertices in $C$.

For any auxiliary vertex $u_{H_{A}}$, the probability that a random
vertex $v_{S,\left(P\mid_{S}\right)}\in C$ is connected to $u_{H_{A}}$
is $2^{-\left|S\cap{\cal F}_{A}\right|}<1/n$, and similarly for $u_{H_{B}}$
and $u_{H}$. Therefore, every auxiliary vertex is connected to less
than a $\left(1/n\right)$-fraction of the vertices in $C$.

\subsection{Soundness}
\begin{lem}
\label{lem:soundness}If the {\sc Label Cover} instance has value
at most $\epsilon^{3}$, then there are no $\left(\beta+\epsilon,\beta\right)$-communities.
\end{lem}

\paragraph{Auxiliary vertices}
\begin{claim}
Every $\left(\beta+\epsilon,\beta\right)$-community does not contain
any auxiliary vertices.\label{claim:auxiliary}\end{claim}
\begin{proof}
There are $\left|V\right|^{2}$ identical copies of each auxiliary
vertex. Since they are identical, any community must either contain
all of them, or none of them. If the community contains all $\left|V\right|^{2}$
copies, then it has a vast majority of auxiliary vertices, so none
of them can have edges to an $\epsilon$-fraction of the community.
\end{proof}

\paragraph{List decoding}
\begin{claim}
The vertices in any $\left(\beta+\epsilon,\beta\right)$-community
$C$ induce at most $4/\epsilon$ different assignments for each variable.\label{claim:list-decoding}\end{claim}
\begin{proof}
Suppose by contradiction that this is not the case. Then, wlog, there
is a line $\left\{ g_{1}\right\} \times{\cal G}$ that receives at
least $2/\epsilon$ different assignments from vertices in $C$. Every
two assignments agree on at most $\left|{\cal F}\right|$ points $\left(g_{1},g'\right)$
on the line, so in total there are at most $2\left|{\cal F}\right|/\epsilon^{2}$
points where at least two assignments agree. Let $R\subseteq{\cal G}$
denote the set of $g'$ such that no two assignments agree on $\left(g_{1},g'\right)$;
we have that $\left|R\right|\geq\left|{\cal G}\right|-2\left|{\cal F}\right|/\epsilon^{2}\geq\left|{\cal G}\right|/2$.
Therefore, by the weak ties property, for at most a $\beta$-fraction
of the vertices $v_{S,\left(q_{\ell}\right)}\in C$, $S\cap R=\emptyset$. 

Consider the remaining $\left(1-\beta\right)$-fraction of vertices
in $C$. Suppose that $v$ assigns a value to some $\left(g_{1},g'\right)$
for $g'\in R$: this value can only agree with one of the $2/\epsilon$
different assignments to $\left(g_{1},g'\right)$. Therefore, in expectation,
each of the $2/\epsilon$ vertices that assign different values for
$\left(g_{1},g'\right)$ is connected to at most a $\left(\beta+\epsilon/2\right)$-fraction
of the vertices in $C$. This is a contradiction to $C$ being a $\left(\beta+\epsilon,\beta\right)$-community.
\end{proof}

\paragraph{Completing the proof}
\begin{proof}
[Proof of Lemma \ref{lem:soundness}] Suppose that at most a $\epsilon^{3}$-fraction
of the {\sc Label Cover} constraints can be satisfied by any single
assignment, and assume by contradiction that $C$ is a $\left(\beta+\epsilon,\beta\right)$-community.
By Claim \ref{claim:list-decoding}, $C$ induces at most $4/\epsilon$
assignments on each variable, so at most $O\left(\epsilon\right)$-fraction
of the constraints are satisfied by any pair of assignments. 

By Markov's inequality, for at least half of the subsets $X_{i}\subset A$,
only an $O\left(\epsilon\right)$-fraction of the constraints that
depend on $X_{i}$ are satisfied. By Claim \ref{claim:auxiliary}
at least $\left(1-\beta\right)$-fraction of the vertices in $C$
assign values to at least one such $X_{i}$. Consider any such vertex
$v_{S,\left(q_{\ell}\right)}$ where $S\ni i$. By construction of
the partitions (Lemma \ref{lem:partition}), each $X_{i}$ shares
approximately the same number of constraints with each $Y_{j}$. Therefore,
for all but an $O\left(\epsilon\right)$-fraction of $Y_{j}$'s, $X_{i}$
and $Y_{j}$ observe a violation - for all the assignments given by
vertices in $C$ to the variables in $Y_{j}$. In other words, $v_{S,\left(q_{\ell}\right)}$
cannot have edges to any vertex $v_{T,\left(r_{\ell}\right)}$ such
that $T\ni j$, for a $\left(1-O\left(\epsilon\right)\right)$-fraction
of $j\in\left[n_{B}/k_{B}\right]$. Finally, applying Claim \ref{claim:auxiliary}
again, at most a $\beta$ fraction of vertices in $C$ do not contain
any of those $j$'s. This is a contradiction to $v_{S,\left(q_{\ell}\right)}$
having edges to $\left(\beta+\epsilon\right)$-fraction of the vertices
in $C$.
\end{proof}
\bibliographystyle{alpha}
\bibliography{communities}

\end{document}